\documentclass[11pt,a4paper]{article}
\usepackage{amsfonts,amssymb,amsthm}
\usepackage{cite} 
\usepackage[T2A]{fontenc}
\usepackage[cp1251]{inputenc}
\usepackage{amsmath}%
\usepackage[unicode]{hyperref}
\usepackage[left=1.5cm,right=1.5cm,
    top=2cm,bottom=3cm,bindingoffset=0cm]{geometry}

\newtheorem{theorem}{Теорема}[section]

\def\pfrac#1#2{\frac{\partial #2}{\partial #1}}

\def\eq#1{\begin{equation}#1\end{equation}}

\def\eqs#1{\begin{equation}\begin{split}#1\end{split}\end{equation}}
\def\seqs#1{\begin{equation*}\begin{split}#1\end{split}\end{equation*}}
\def\seq#1{\begin{equation*}#1\end{equation*}}

\allowdisplaybreaks
\numberwithin{equation}{section}

\title{\bf On the symmetry classification of integrable chains in 3D. Darboux-integrable reductions and their higher symmetries}
\author{\bf R.N.~Garifullin, I.T.~Habibullin, 
}

\date{} 

\begin{document}
\maketitle


\abstract {This paper proposes a method for identifying and classifying integrable nonlinear equations with three independent variables, one of which is discrete and the other two are continuous. A characteristic property of this class of equations, called Toda-type chains, is that they admit finite-field reductions in the form of open chains with enhanced integrability. The paper results in a theorem stating that all known integrable Toda-type chains admit reductions in the form of an open chain of length three with a family of second-order evolutionary type symmetries. Apparently, this property of Toda-type chains can be used as an effective classification criterion when compiling lists of integrable differential-difference equations in 3D. } 

\vspace{0.5cm}


\vspace{0.5cm}

\large

\section{Introduction}

The paper discusses the problem of creating an efficient classification algorithm for integrable equations with three independent variables of the following form:
\begin{equation} \label{todatype}
u^n_{xy} = f(u^{n+1},u^{n},u^{n-1}, u_{x}^n,u_{y}^n).
\end{equation}
It is well known that the symmetry approach developed in the works of \cite{AdlerShabatYamilov, IbragimovShabat, MikhailovShabatYamilov, SvinolupovSokolov, ZhiberShabat, SokolovBook20, ISh, GarifullinYamilovLevi17, GarifullinYamilovLevi} and proven to be a highly effective tool for classifying integrable models in 1+1 dimensions encounters significant difficulties when applied to equations with three independent variables (see, for example, \cite{MikhailovYamilov1998}). This is due to the problem of nonlocal variables involved in the symmetries of equations in 3D. Therefore, the search for alternative approaches that take into account the specifics of three-dimensional models is particularly relevant. The literature contains a series of works devoted to the study of this problem (see, for example, \cite {MikhailovYamilov1998,FerapontovKhusnutdinova2004,ShabatYamilov,XenitidisNijhoffLobb2011,RamaniGrammaticosBountis1989,AdlerBobenkoSuris2009}).

Examples of integrable three-dimensional lattices of Toda type (see  \eqref{todatype}) appeared as early as the 19th century in the works of Laplace, in connection with his cascade method for integrating linear hyperbolic equations. Later, in the works of Darboux, Goursat,  Moutard and others, it was discovered that finite systems of nonlinear hyperbolic equations, related to an infinite system (called  in modern terms the Toda lattice), are solved explicitly \cite{Darboux, GanzhaTsarev2001}. Interest in nonlinear equations of the form \eqref{todatype} was revived with the discovery of the inverse scattering transform method. A connection was revealed between classical results on methods of integrating partial differential equations and modern integrability theory (see \cite{Mikhailov79, Mikhailov81, ShabatYamilov81,LSS83}). It should be noted that the property of having two-dimensional reductions with increased integrability, discovered in classical works, is inherent not only to the Toda lattice, but also to all other known integrable chains of the form \eqref{todatype}, see \cite{HabibullinKuznetsova20}.

The idea of investigating a multidimensional integrable partial differential equation by describing the characteristic properties of its reductions with a smaller number of independent variables is well known and is actively used in applications \cite{AblowitzSegur}. A similar approach was used in the works 
\cite{Habibullin2013, HabibullinPoptsova18, Kuznetsova19, HabibullinKuznetsova20} to classify quasilinear integrable nonlinear 3D equations of the form \eqref{todatype}. There, the presence of two-dimensional Darboux-integrable reductions of arbitrary order is used as a classification criterion. More precisely, it is required that the finite-field system of hyperbolic equations obtained by imposing suitable cutoff conditions at two arbitrary points $n_1$ and $n_2$ of the  chain \eqref{todatype}  is integrable in the Darboux sense, i.e. it admits a complete sets of characteristic integrals in each of the $x$ and $y$ directions. Such reductions, called open chains, are well known, see for example \cite{ShabatYamilov81}, our observation is that all known integrable equations of the form \eqref{todatype} admit reductions in the form of open chains of arbitrary order, and all of them are Darboux integrable.
To identify Darboux-integrable systems, a convenient algebraic criterion  is used introduced in 1980 by A.B. Shabat. It is based on the concept of a characteristic Lie algebra. The criterion asserts that a system of hyperbolic equations is Darboux-integrable if and only if both of its characteristic Lie algebras, corresponding to the characteristic directions, have finite dimension (see \cite{ShabatYamilov81,	ZMHS2012}).

For quasilinear chains, i.e., chains of the following form:
\begin{equation}\label{0}
u_{n,xy}=A_1u_{n,x}u_{n,y}+A_2u_{n,x}+A_3u_{n,y}+A_4,
\end{equation}
here the coefficients are smooth functions of the dynamical variables $A_i=A_i(u_{n+1},u_n,u_{n-1})$ for $i=1,2,3,4$,
the above-mentioned classification problem was completely solved in the works \cite{Habibullin2013, HabibullinPoptsova18, Kuznetsova19, HabibullinKuznetsova20}. We present the result obtained in these works. Any  equation of class \eqref{0}, integrable in the sense of the definition given above, can be reduced by point transformations to one of the equations of the list
\begin{itemize}
\item[$1)$] $u_{n,xy} = e^{u_{n+1} - 2 u_n + u_{n-1} },$
\item[$2)$] $u_{n,xy} = e^{u_{n+1}} - 2 e^{u_n} + e^{u_{n-1}},$
\item[$3)$] $u_{n,xy} = e^{u_{n+1}-{u_n}} -  e^{u_n-u_{n-1}},$
\item[$4)$] $u_{n,xy} = \left(u_{n+1} - 2 u_n + u_{n-1}  \right) u_{n,x}, $
\item[$5)$] $u_{n,xy} = \left(e^{u_{n+1}-{u_n}} -  e^{u_n-u_{n-1}}\right)u_{n,x},$
\item[$6)$] $u_{n,xy}=\alpha_nu_{n,x}u_{n,y}, \quad \alpha_n = \frac{1}{u_n - u_{n-1}} - \frac{1}{u_{n+1}-u_n}$
\end{itemize}

In the aforementioned articles \cite{Habibullin2013, HabibullinPoptsova18, Kuznetsova19, HabibullinKuznetsova20} the authors provided a list containing one more equation:
\begin{equation}\label{wrong}
u_{n,xy} = \alpha_n(u_{n,x} + u^2_n - 1)(u_{n,y} + u^2_n - 1) - 2 u_n(u_{n,x}+u_{n,y}+u^2_n - 1),
\end{equation}
However, it turned out that this equation can be reduced to (6) by a point change of the variables. This is shown below in this paper (see (3.15), (3.17)). Note that chains (1)-(3) were obtained and studied in the classical works of Darboux, Goursat, and Moutard. Chains 4) and 5) apparently were found in the paper \cite{ShabatYamilov}. Equation 6) was found simultaneo\-usly and independently in the papers \cite{ShabatYamilov} and \cite{Ferapontov97}. A discussion of the irreversible substituti\-ons that relate the equations of the list can be found in \cite{ShabatYamilov}.

We emphasize that when considering lattices of general form, the integrability conditions derived from the finite-dimensionality property of the characteristic algebra become extremely complex. For this reason, we were unable to solve the problem of describing integrable cases of a general chain \eqref{todatype} using the method of characteristic algebras. To date, the problem remains unresolved.  Therefore, we  improve the two-stage procedure for classifying integrable chains proposed in \cite{Habibullin2013, HabibullinPoptsova18, Kuznetsova19, HabibullinKuznetsova20} as follows. The first stage, which consists of using Darboux-integrable reductions, remains unchanged. However, when identifying the property of enhanced integrability, we propose using Theorem 3.1. It states that for all known integrable cases of the Toda type lattices, the over mentioned special reductions have second-order symmetries. We propose to use the requirement of the presence of second-order symmetries for finite-field reductions of the type under consideration as a classification criterion. This will allow one to apply the well-developed technique of the symmetry approach for classification.

\section{Special (locking) boundary conditions for  integrable 3D models and open chain type reductions.}

The dynamics due to chains 4) and 6) from the list above are completely consistent with the constraints of the form
\begin{equation}\label{constraint}
u_{n_0} = C,
\end{equation}
where $C$ is an arbitrary constant. By imposing constraints of this type at two different points $n_1$ and $n_2$ on the integer line, we reduce any of these two chains to three independent systems defined on two semi-infinite segments $(-\infty, n_1-1]$, $[n_2+1, +\infty)$ and on a finite integer segment $[n_1,n_2]$. In some sources, constraints with such properties are called locking boundary conditions. The last of these systems is an open chain, integrable in the Darboux sense. For equation 4), the open chain has the form
\eqs{&u_{n_1}=C_1,\\ &u_{nxy}=(u_{n+1}-2u_{n}+u_{n-1})u_{nx},\qquad n_1+1\leq n\leq n_2-1,\\ &u_{n_2}=C_2.}
Similarly, for 6), we have
\eqs{&u_{n_1}=C_1,\\ &u_{nxy}=\alpha_nu_{nx}u_{ny},\qquad n_1+1\leq n\leq n_2-1,\\ &u_{n_2}=C_2.}
The remaining equations in the list are not consistent with the constraint \eqref{constraint}. For them, the locking boundary condition is associated with the point at infinity $u_{n_0}=\infty$, at which the right-hand sides of equations 1)--3) and 5) have an essential singularity.

Now we derive an open chain corresponding to equation 1). By changing the variables $u_{n_1}=v_{n_1}+2\ln \varepsilon$, $u_{n_2}=v_{n_2}+2\ln \varepsilon$ and also
$u_{n}=v_{n}+\ln \varepsilon$, for $n\neq n_1$, $n\neq n_2$, we reduce chain 1) to the following form
\eqs{&v_{n_1,xy}=\varepsilon e^{v_{n_1+1}-2v_{n_1}+v_{n_1-1}},\\ &v_{nxy}=e^{v_{n+1}-2v_{n}+v_{n-1}},\quad n_1+1\leq n\leq n_2-1,\\ &v_{n_2,xy}=\varepsilon e^{v_{n_2+1}-2v_{n_2}+v_{n_2-1}}.}
Let us assume that the parameter $\varepsilon$ tends to zero. In the limit, we obtain the desired open chain 
\eqs{&v_{n_1,xy}=0,\\ &v_{nxy}=e^{v_{n+1}-2v_{n}+v_{n-1}},\qquad n_1+1\leq n\leq n_2-1,\\ &v_{n_2,xy}=0. \label{2.5}}

Using this technique, we can obtain open-chain type reductions for equations 2), 3), 5). Omitting the calculations, we present only the corresponding open chains. For 2), we have
\eqs{&v_{n_1,xy}= e^{v_{n_1+1}}-2e^{v_{n_1}},\\ &v_{n,xy} = e^{v_{n+1}} - 2 e^{v_n} + e^{v_{n-1}},\quad n_1+1\leq n\leq n_2-1,\\ &v_{n_2,xy}=-2 e^{v_{n_2}}+e^{v_{n_2-1}}. \label{2.6}}

For 3) and 5) we have, respectively,
\eqs{&v_{n_1,xy}= e^{v_{n_1+1}-v_{n_1}},\\ &v_{n,xy} = e^{v_{n+1}-{v_n}} -  e^{v_n-v_{n-1}},\quad n_1+1\leq n\leq n_2-1,\\ &v_{n_2,xy}=- e^{v_{n_2}-v_{n_2-1}},}

\eqs{&v_{n_1,xy}= e^{v_{n_1+1}-v_{n_1}}v_{n_1,x},\\ &v_{n,xy} = \left(e^{u_{n+1}-{u_n}} -  e^{u_n-u_{n-1}}\right)u_{n,x},\quad n_1+1\leq n\leq n_2-1,\\ &v_{n_2,xy}=- e^{v_{n_2}-v_{n_2-1}}v_{n_2,x}.}

\section{Symmetries of open chains.}

Note that all of the open chains considered above admit complete sets of independent characteristic integrals, i.e., they are integrable in the Darboux sense, \cite{HabibullinKuznetsova20}. It is well known (see \cite{ZhiberSokolov2001}) that Darboux-integrable equations admit a wide class of higher symmetries, since there exists a mapping that takes integrals to symmetries. Therefore, the open chains listed also admit higher symmetries as well. 
Below, we prove the following assertion for equations 1)-6), and separately show that equation \eqref{wrong} reduces to 6) by a point substitution.

\begin{theorem} Open chains of length three corresponding to any of the known integrable Toda-type equations (see \eqref{todatype}) admit higher  symmetries of the order 2, that linearly depends on the second order derivaties.
\end{theorem}
\begin{proof}
Consider the open chains of length three obtained from the equations in the list \eqref{todatype}. We introduce the vector $U=(p,q,r)^T$, whose coordinates are related to the dynamical variables $u_n$ as follows $p=u_{-1},q=u_0,r=u_1$. To prove Theorem 1.4, it suffices to present for each of the open chains of length three
a higher symmetry of the form
\eq{U_t=A(U,U_x)U_{xx}+B(U,U_x),\label{symg}}
where $A(U,U_x)$ is a square matrix of size $3\times 3$, $B(U,U_x)$ is a three-dimensional vector. The symmetry search algorithm is discussed in detail using chain 5) as an example in the Appendix (see below).

1) It follows from \eqref{2.5} that the open chain of length three for the first equation in the list has the form
\eqs{&p_{xy}=0,\\ &q_{xy}=e^{p+r-2q},\\ &r_{xy}=0.}
The sought symmetry is a system of second-order evolution equations depending on the functional parameters

\eqs{&p_t=W a_1(p_x,r_x)+p_{xx}a_2(p_x,r_x)+r_{xx}a_3(p_x,r_x)+a_4(p_x,r_x)+c_1p+c_2r,\\ &(2q-p-r)_t=((2q-p-r)_x-D_x)a_5(p_x,r_x),\\ &r_t=Wa_6(p_x,r_x)+p_{xx}a_7(p_x,r_x)+r_{xx}a_8(p_x,r_x)+a_9(p_x,r_x)+c_3p+c_4r, \\ &\quad W=q_{xx}+\frac14((2q_x-p_x-r_x)^2.}
Here $a_i$ is an arbitrary function of the variables $p_x,r_x$; $W$ is a second-order integral, and the parameters $c_i$ are arbitrary constants.

2) It follows from \eqref{2.6} that the required reduction for the second equation in the list is of the form
\eqs{&p_{xy}=e^q-2e^p,\\ &q_{xy}=e^r-2e^q+e^p,\\ &r_{xy}=-2e^r+e^q.}
It can be verified that the following system of evolutionary type equations 
\eqs{&p_t=c_1(p_x+D_x)(p_x+2q_x+r_x)+c_2p,\\
&q_t=c_1(r_{xx}-p_{xx}+q_x(r_x-q_x))+c_2q,\\ 
&r_t=-c_1(r_x+D_x)(p_x+2q_x+r_x)+c_2r}
provides a symmetry for the open chain above.

3) Next we present a reduction of length three for the third equation of the list and its symmetry
\eqs{&p_{xy}=e^{q-p},\\ &q_{xy}=e^{r-q}-e^{q-p},\\ &r_{xy}=-e^{r-q}.}
\eqs{&p_t=W_2a_1(W_1)+D_xa_2(W_1)+a_3(W_1)+(p_x-D_x)a_4(W_1)\\&\quad+(p+q+r)c_1+(q_xr_x+r_{xx})c_2\\
&q_t=W_2a_1(W_1)+D_x a_2(W_1)+a_3(W_1)+q_xa_4(W_1)+(p+q+r)c_1+p_xr_xc_2\\
&r_t=W_2a_1(W_1)+D_xa_2(W_1)+a_3(W_1)+(r_x+D_x)a_4(W_1)\\&\quad+(p+q+r)c_1+(p_xq_x-q_{xx})c_2,} where $a_1,a_2,a_3,a_4$ are arbitrary functions of the first-order integral $W_1=p_x+q_x+r_x$, $c_1,c_2$ are constants, and the function $W_2$ is a second-order integral. It has the form:
\seq{W_2=\frac{1}{3}(r_{xx}-p_{xx}+p_xq_x+p_xr_x+q_xr_x).}

4) Let us discuss the properties of the fourth chain reduction
\eqs{&p_{xy}=(C_1-2p+q)p_{x},\\ &q_{xy}=(p-2q+r)q_{x},\\ &r_{xy}=(q-2r+C_2)r_{x}.}
The corresponding  symmetry is as follows
\eqs{&p_t=c_1p_x\left(\frac{p_{xx}}{p_x}+2\frac{q_{xx}}{q_{x}}+\frac{r_{xx}}{r_{x}}\right)+c_2p_x,\\ &q_t=c_1q_x\left(-\frac{p_{xx}}{p_x}+\frac{r_{xx}}{r_{x}}\right)+c_2q_x,\\ &r_t=c_1r_x\left(-\frac{p_{xx}}{p_x}-2\frac{q_{xx}}{q_{x}}-\frac{r_{xx}}{r_{x}}\right)+c_2r_x.}

5) Below we present the symmetries of the fifth chain reduction
\eqs{&p_{xy}=e^{q-p}p_x,\\ &q_{xy}=(e^{r-q}-e^{q-p})q_{x},\\ &r_{xy}=-e^{r-q}r_{x},\label{3.10}} corresponding to each of the characteristic directions. We begin with the symmetry in the $x$ direction.
\eqs{&p_t=p_x \left(D_x a_1(W_1)+W_2a_2(W_1)+\left(\frac{q_{xx}}{q_x}+\frac{r_{xx}}{r_x}-p_x+D_x\right)a_3(W_1)+a_4(W_1)\right)+c_1,\\
&q_t=q_x \left(D_x a_1(W_1)+ W_2a_2(W_1)+\left(\frac{r_{xx}}{r_x}-\frac{p_{xx}}{p_x}-q_x\right)a_3(W_1)+a_4(W_1)\right)+c_1,\\
&r_t=r_x \left(D_x a_1(W_1)+W_2a_2(W_1)-\left(\frac{p_{xx}}{p_x}+\frac{q_{xx}}{q_x} +r_x+D_x\right)a_3(W_1)+a_4(W_1)\right)+c_1.\label{5x}}
Here $a_1,a_2,a_3$ are arbitrary functions of the first-order integral $W_1=\ln p_xq_xr_x$, parameter $c_1$ is an arbitrary constant, and the second-order integral $W_2$ is of the form:
\seq{W_2=\frac{p_{xx}}{p_x}-\frac{r_{xx}}{r_x}+p_x+q_x+r_x.}

This system is not invariant under the involution $x\leftrightarrow y$. Therefore, we separately search for its symmetry in the $y$ direction.
\eqs{
&p_\tau=D_y a_1(W_3)+W_4a_2(W_3)+(p_y-D_y)a_3(W_3)+a_4(W_3)\\&\quad+c_1(r_{yy}+r_y(q_y+e^{r-q}+e^{q-p})+e^{r-p}),\\
&q_\tau=D_y a_1(W_3)+W_4a_2(W_3)+q_ya_3(W_3)      +a_4(W_3)\\&\quad+c_1(p_yr_y+p_ye^{r-q}+r_ye^{q-p}+e^{r-p}),\\
&r_\tau=D_y a_1(W_3)+W_4a_2(W_3)+(q_y+D_y)a_3(W_3)+a_4(W_3)\\&\quad+c_1(-p_{yy}+p_y(q_y+e^{r-q}+e^{q-p})+e^{r-p}).}
Here $a_1,a_2,a_3,a_4$ are arbitrary functions of the first order integral  $$W_3=p_y+q_y+r_y+e^{r-q}+e^{q-p},$$ $c_1$ is an arbitrary constant, function  $W_4$ is a second order integral:
\seq{W_4=r_{yy}-p_{yy}+p_yq_y+p_yr_y+q_yr_y+(e^{q-p}+e^{r-q})(p_y+r_y)+e^{r-p}.}
6) Let us proceed with the last equation of the list
\eqs{&p_{xy}=\left(\frac{1}{p-C_1}-\frac{1}{q-p}\right)p_xp_y,\\ &q_{xy}=\left(\frac{1}{q-p}-\frac{1}{r-q}\right)q_xq_y,\\ &r_{xy}=\left(\frac{1}{r-q}-\frac{1}{C_2-r}\right)r_xr_y.}
The following system is its symmetry
\eqs{&p_t=p_x D_x a_1(W_1)+p_x W_2a_2(W_1)+p_x \left( \left(\ln \frac{p_x}{W_1(p-C_1)^2}\right)_x -D_x\right)a_3(W_1)\\&\qquad+p_x a_4(W_1)+(p-C_1)(p-C_2)c_1,\\
     &q_t=q_x D_x a_1(W_1)+q_x W_2a_2(W_1)+q_x \left(\ln \frac{p_x(r-q)(r-C_2)}{r_x(p-C_1)(q-p)}\right)_x a_3(W_1) \\&\qquad    +q_x a_4(W_1)+(q-C_1)(q-C_2)c_1,\\
     &r_t=r_x D_x a_1(W_1)+r_x W_2a_2(W_1)+r_x\left(\left(\ln\frac{W_1(r-C_2)^2}{r_x}\right)_x+D_x\right)a_3(W_1)\\&\qquad+r_xa_4(W_1)+(r-C_1)(r-C_2)c_1.}
Here $a_1,a_2,a_3,a_4$ are arbitrary functions of the first-order integral $$W_1=\frac{p_xq_xr_x}{(p-C_1)(q-p)(q-r)(r-C_2)},$$ $c_1$ is an arbitrary constant. The second order integral  $W_2$ is given by:
\seq{W_2=\frac{p_{xx}}{p_x}-\frac{r_{xx}}{r_x}-\frac{2p_x}{p-C_1}+q_x\left(\frac1{q-r}-\frac1{q-p}\right)+\frac{2r_x}{r-C_2}.}

\end{proof}
Now we will show that lattice \eqref{wrong} reduces to the sixth equation of the list \eqref{todatype}.
To this end we note that replacing
\eq{u_n=\frac{1+v_n}{1-v_n}} in equation \eqref{wrong} reduces it to the form:
\eqs{v_{n,xy}=\beta_n(v_{n,x}+2v_n)(v_{n,y}+2v_n)-4v_n-2v_{n,x}-2v_{n,y},\\ \beta_n=\frac{1}{v_n-v_{n-1}}-\frac{1}{v_{n+1}-v_n}.}
Then by the non-autonomous replacement
\eq{v_n=e^{-2x-2y}w_n} one converts the latter to the equation 6) for the function $w$.

\section*{Appendix. Evaluation of symmetry \eqref{5x}.}
\setcounter{equation}{0}
\renewcommand{\theequation}{A.\arabic{equation}}

Using the example of constructing a symmetry \eqref{5x} of the system \eqref{3.10}, we illustrate in detail the algorithm for finding symmetries of open chains. We represent the system \eqref{3.10} as \eq{U_{xy}=F(U)U_{x},\label{eq5},} where $F(U)=\hbox{diag}(F_1,F_2,F_3)=\hbox{diag}(e^{q-p},e^{r-q}-e^{q-p},-e^{r-q})$ is a diagonal matrix. We recall  that the $y$-integrals of this equation, i.e. non-trivial solutions to the equations $D_yW_1=0$ and $D_yW_2=0$ have the form \seq{W_1=\ln p_xq_xr_x,}
\seq{W_2=\frac{r_{xx}}{r_x}-\frac{p_{xx}}{p_x}+p_x+q_x+r_x.}
The compatibility condition of the equations \eqref{symg} and \eqref{eq5} yields \eq{(F(U)U_{x})_t-(A(U,U_x)U_{xx}+B(U,U_x))_{xy}=0,\label{usl5}.} We open the brackets in the resulting equality and collect the coefficients in front of the highest derivative $U_{xxx}$ to get:
\eq{FA-D_y(A)-AF=0.\label{gl}} Since $F$ and $A$ do not contain the variable $U_y$, while the middle term in the last equation is of the form $D_y(A)=\sum_{i=1}^3\left(\pfrac{U_i}{A}U_{i,y}+\pfrac{U_{i,x}}{A}F_iU_{i,x}\right)$, we conclude that $A$ should not depend on $U$.
Since $F$ is a diagonal matrix, the equality \eqref{gl} in element-wise notation has the form \eq{(F_i-F_j)A_{i,j} -\sum_{k=1}^3\pfrac{U_{k,x}}{A_{ij}}F_kU_{k,x}=0,\label{sys_A}} where $i,j$ are the row and column numbers, respectively. The solution to system \eqref{sys_A} is:
\eq{A=\left(\begin{matrix}
A_{11}(W_1) & A_{12}(W_1)\frac{p_x}{q_x} & A_{13}(W_1)\frac{p_x}{r_x}\\
A_{21}(W_1)\frac{q_x}{p_x} & A_{22}(W_1) & A_{13}(W_1)\frac{q_x}{r_x}\\
A_{31}(W_1)\frac{r_x}{p_x} & A_{32}(W_1)\frac{r_x}{q_x} & A_{33}(W_1)
\end{matrix}\right),} where the entries $A_{ij}(W_1)$ are arbitrary functions of  $W_1.$
Applying the operators $\pfrac{U_{j,xx}\partial U_{k,y}}{^2},$ $j,k=1,2,3$ to the equation \eqref{usl5} yields the system \eq{\pfrac{U_{k}\partial U_{j,x}}{^2B}=0,} which shows that the vector function $B$ decomposes into sums of two functions:
\eq{B=B^{(1)}(U)+B^{(2)}(U_x).}
Application of the operators $\pfrac{U_{j,x}\partial U_{k,y}}{^2},$ $j,k=1,2,3$ to the equation \eqref{usl5} leads to the system \eq{\pfrac{U_{k}\partial U_{j}}{^2B}=0.} The latter  convinces that function $B^{(1)}$ linearly depends on in each of its arguments:
\eq{B^{(1)}=B^{(p)}p+B^{(q)}q+B^{(r)}r.}
Thus, we can conclude that equation \eqref{usl5} splits down into two parts:
\seq{e^{q-p}P_1(p,q,r,p_x,q_x,r_x,p_{xx},q_{xx},r_{xx})+e^{r-q}P_2(p,q,r,p_x,q_x,r_x,p_{xx},q_{xx},r_{xx})=0.}
Vector-functions  $P_1$ and $P_2$ linearly depend on some of their arguments $p,q,r,p_{xx},q_{xx},r_{xx}$. Since linear and exponential functions are linearly independent, we require functions $P_1$ and $P_2$ to be equal to zero.
Comparing the coefficients in front of the first powers of the variables $p,q,r$, we find the vectors  $B^{(p)},B^{(q)},B^{(r)}$: $$B^{(p)}=C_p (1,1,1)^T,\quad B^{(q)}=C_q (1,1,1)^T,\quad B^{(r)}=C_r (1,1,1)^T, $$ where $C_p,C_q,C_r$ are some constants.

Let us return to the equation obtained above for finding $P_1$ and $P_2$. Equating the coefficients of the variables $p_{xx}$, $q_{xx}$, $r_{xx}$ in this equation generates systems of the first- and second-order partial differential equations with respect to the components of the vector $B^{(2)}$. The compatibility conditions of this system lead to: 
\eqs{&A_{21}(W_1)=A_{23}(W_1)-2A_{12}(W_1)+2A_{11}(W_1),\\ &A_{22}(W_1)=A_{23}(W_1)-A_{12}(W_1)+A_{11}(W_1),\\ &A_{23}(W_1)=A_{11}'(W_1)-A_{12}'(W_1)+A_{13}'(W_1)+A_{13}(W_1).} The solution $B_1^{(2)}$ is of the form:
\eq{B_1^{(2)}=p_x(r_x+q_x)(A_{12}(W_1)-A_{13}(W_1))+(A_{11}(W_1)-A_{12}(W_1))p_x^2+p_xa_4(W_1)+c_1,} where $a_4$ is an arbitrary function of the variable $W_1$ and $c_1$ is an arbitrary constant.
Since the first component of the vector $P_2$ vanishes, we find: $$C_r=C_q=0$$ and
$$B_2^{(2)}=q_x^2(A_{11}(W_1)-A_12(W_1))+q_x(p_x+r_x)(A_{12}(W_1)-A_{13}(W_1))+q_xa_4(W_1)+c_1.$$

Next, we move on to the second components of the vectors $P_1,P_2$. They have the following form:
$$P_i=P_{ip}p_{xx}+P_{iq}q_{xx}+P_{ir}r_{xx}+P_{i0}.$$
From the requirement that $P_{10}^{(2)}$ is equal to zero, we obtain $C_p=0$. From the requirement that $P_{2p}^{(2)}, P_{2q}^{(2)}, P_{2r}^{(2)}$ is equal to zero, we find:
\seqs{&A_{31}(W_1)=2A_{11}'(W_1)-4A_{12}'(W_1)+2A_{13}'(W_1)+A_{13}(W_1)-2A_{12}(W_1)+2A_{11}(W_1),\\ &A_{32}(W_1)=A_{13}(W_1)-A_{12}(W_1)+A_{31}(W_1),\\ &A_{33}(W_1)=A_{13}(W_1)-A_{11}(W_1)+A_{31}(W_1).} Equating $P_{20}^{(2)}$ to zero we obtain:
\seq{B_3^{(2)}=r_x^2(A_{11}(W_1)-A_{12}(W_1))+r_x(p_x+q_x)(A_{12}(W_1)-A_{13}(W_1))+a_4(W_1)+c_1.} 
Now it is easy to verify that the third components of the vectors $P_1$, $P_2$ are identically zero.
Renaming the functions as follows:
\seqs{&A_{11}(W_1)=a_1'(W_1)-a_2(W_1)+a_3'(W_1),\\&A_{12}(W_1)=a_1'(W_1)+a_3(W_1)+a_3'(W_1),\\ &A_{13}(W_1)=a_1'(W_1)+a_2(W_1)+a_3(W_1)+a_3'(W_1)} we obtain the formula \eqref{5x}.

\section{Conclusions.}

 It was shown in the papers \cite{Habibullin2013, HabibullinPoptsova18, Kuznetsova19, HabibullinKuznetsova20}, that all known integrable equations of the form \eqref{todatype} admit Darboux-integrable finite-field reductions, which can be obtained by terminating the chain by imposing special boundary conditions. To verify Darboux-integrability, these works used a rather complex algebraic criterion, which, in the case of a general-position system, leads to excessively cumbersome integrability conditions.

In this paper, we propose a simpler, but highly effective, symmetry criterion for identifying integrable reductions. It is based on Theorem 3.1, which states that for all known integrable two-dimensional Toda chain equations, the corresponding special finite-field reductions of length three admit second-order symmetries.

\end{document}